\documentclass[times]{qjrms4}%
\usepackage{amsmath}
\usepackage{amsfonts}
\usepackage{amssymb}
\usepackage{graphicx}%
\providecommand{\U}[1]{\protect\rule{.1in}{.1in}}
\def\volumeyear{2013}
\newtheorem{algorithm}{Algorithm}
\newtheorem{theorem}{Theorem}
\newtheorem{remark}{Remark}
\newenvironment{proof}{Proof.}{$\Box$}
\begin{document}
%
\runningheads{J.~Mandel, E.~Bergou, S.~Gratton}%
{4DVAR by ensemble Kalman smoother}
\title{4DVAR by ensemble Kalman smoother}
\author{Jan Mandel,\affil{a}\corrauth\
Elhoucine Bergou,\affil{b} and
Serge Gratton\affil{b}}
\address{\affilnum{a}
University of Colorado Denver, Denver, CO, USA\\
\affilnum{b}INP-ENSEEIHT and CERFACS, Toulouse, France}
\corraddr{Department of Mathematical and Statistical Sciences, \\
University of Colorado Denver, Denver, CO 80217-3364, USA.\\
E-mail: jan.mandel@gmail.com}
\begin{abstract}
In this paper, we propose to use the ensemble Kalman smoother (EnKS) as linear least squares solver in the
Gauss-Newton method for the large nonlinear least squares in incremental 4DVAR.
The ensemble approach is naturally parallel over the ensemble members and no tangent or adjoint operators
are needed. Further, adding a regularization term results in replacing the Gauss-Newton method, which may diverge, by
the Levenberg-Marquardt method, which is known to be convergent.
The regularization is implemented efficiently as an additional observation in the EnKS.
\end{abstract}
\maketitle

\section{Introduction}

4DVAR is a dominant data assimilation method used in weather forecasting
centers worldwide. 4DVAR attempts to reconcile model and data variationally,
by solving a very large weighted nonlinear least squares problem. The unknown
is a vector of system states over discrete points in time, when the data are
given. The objective function minimized is the sum of the squares of the
differences of the initial state from a known background state at the initial
time and the differences of the values of observation operator and the data at
every given time point. In the weak-constraint 4DVAR
\citep{Tremolet-2007-ME4}%
, considered here, the model error is accounted for by allowing the ending and
starting state of the model at every given time point to be different, and
adding to the objective function also the sums of the squares of those
differences. The sums of the squares are weighted by the inverses of the
appropriate error covariance matrices, and much of the work in the
applications of 4DVAR goes into modeling those covariance matrices.

In the incremental approach
\citep{Courtier-1994-SOI}%
, the nonlinear least squares problem is solved iteratively by using a
succession of linear least square solutions. The major cost in 4DVAR
iterations is in evaluating the model, tangent and adjoint operators, and
solving large linear least squares. A significant software development effort
is needed for the additional code to implement the tangent and adjoint
operators to the model and the observation operators. Straightforward
linearization, called the incremental approach
\citep{Courtier-1994-SOI}%
, leads to the Gauss-Newton method for nonlinear least squares
\citep{Bell-1994-IKS,Tshimanga-2008-LPA}%
. However, Gauss-Newton iterations may not converge, not even locally.
Finally, while the evaluation of the model operator is typically parallelized
on modern computer architectures, there is a need to further parallelize the
4DVAR process itself.

The Kalman filter is a sequential Bayesian estimation of the gaussian state of
a linear system at a sequence of discrete time points. At each of the time
points, the use of the Bayes theorem results in an update of the state,
represented by its mean and covariance. The Kalman smoother simply considers
all states at all time points from the beginning to be a large composite
state. Consequently, the Kalman smoother is obtained from the Kalman filter by
simply applying the same update as in the filter to the past states as well.
However, historically, the focus was on efficient short recursions
\citep{Rauch-1965-MLE,Strang-1997-LAG}%
, similar the sequential Kalman filter.

It is well known that weak constraint 4DVAR is equivalent to the Kalman
smoother in the linear case. To apply the Kalman smoother in the nonlinear
case, the problem needs to be linearized, leading to variants of the extended
Kalman filter and the Gauss-Newton method. Use of the Kalman smoother to solve
the linear least squares in the Gauss-Newton method is known as the iterated
Kalman smoother, and considerable improvements can be obtained against running
the Kalman smoother only once
\citep{Bell-1994-IKS,Fisher-2005-EKS}%
.

The Kalman filter and smoother require maintaining the covariance of the
state, which is not feasible for large systems, such as in numerical weather
prediction. Hence, the ensemble Kalman filter (EnKF) and ensemble Kalman
smoother (EnKS)
\citep{Evensen-2009-DAE}
use a Monte-Carlo approach for large systems, representing the state by an
ensemble of simulations, and estimating the state covariance from the
ensemble. The implementation of the EnKS in \cite{Stroud-2010-EKF} uses the
adjoint model with the short recursions as in the KS. However, the
implementations in \cite{Khare-2008-IAE,Evensen-2009-DAE} do not depend on the
adjoint model and simply apply EnKF algorithms to the composite state over
multiple time points. We use the latter approach here.

The EnKF has become a competitive method for data assimilation. Consequently,
combinations of ensemble and variational approaches have become of
considerable recent interest. Estimating the background covariance for 4DVAR
from an ensemble was one of the first connections
\citep{Hamill-2000-HEK-x}%
, and it is now standard and became operational
\citep{Wang-2010-IEC}%
. Gradient methods in the span of the ensemble for one analysis cycle (i.e.,
3DVAR) include \cite{Zupanski-2005-MLE}, \cite{Sakov-2012-IES} (with square
root EnKF as a linear solver in Newton method), and \cite{Bocquet-2012-CII},
who added regularization and use LETKF-like approach to minimize the nonlinear
cost function over linear combinations of the ensemble.
\cite{Liu-2008-EFV,Liu-2009-EFV} combine ensembles with (strong constraint)
4DVAR and minimize in the observation space. Their method, called Ens4DVAR,
does not need tangent or adjoint operators also. \cite{Zhang-2009-CEK} use a
two-way connection between EnKF and 4DVAR, to obtain the covariance for 4DVAR,
and 4DVAR to feed the mean analysis into EnKF. However, ensemble methods for
the solution of the 4DVAR \emph{nonlinear} least squares problem itself or for
the weak constraint 4DVAR do not seem to have been developed before.

In this paper, we propose to use the ensemble Kalman smoother (EnKS) as linear
least squares solver in 4DVAR. The ensemble approach is naturally parallel
over the ensemble members. The rest of the computational work is relatively
cheap compared to the ensemble of simulations, and parallel dense linear
algebra libraries can be used. The proposed approach uses finite differences
from the ensemble, and no tangent or adjoint operators are needed. To
stabilize the method and assure convergence, a Tikhonov regularization term is
added to the linear least squares, and the Gauss-Newton method becomes the
Levelberg-Marquardt method. The Tiknonov regularization is implemented within
EnKS as a computationally cheap additional observation
\citep{Johns-2008-TEK}%
. We call the resulting method EnKS-4DVAR.

The paper is organized as follows. In Section \ref{sec:4dvar}, we review the
formulation of 4DVAR. The EnKS for the incremental linearized squares problem
is reviewed in Section \ref{sec:EnKS}. The new method without tangent
operators is introduced in Section \ref{sec:FD}. The modifications for the
regularization and the Levenberg-Marquardt method are presented in Section
\ref{sec:reg}. Section \ref{sec:comp} contains the results of our
computational experiments, and Section \ref{sec:conclusion} is the conclusion.

\section{Incremental 4DVAR and the Gauss-Newton method}

\label{sec:4dvar}We want to determine $x_{0},\ldots,x_{k}$, where $x_{i}$ is
the state at time $i$, from the background state,%
\[
x_{0}\approx x_{\mathrm{b}},
\]
the model,%
\[
x_{i}\approx\mathcal{M}_{i}\left(  x_{i-1}\right)  ,
\]
and the observations%
\[
\mathcal{H}_{i}\left(  x_{i}\right)  \approx y_{i},
\]
where $\mathcal{M}_{i}$ is the model operator, and $\mathcal{H}_{i}$ is the
observation operator. Quantifying the uncertainty by covariances, with
$x_{0}\approx x_{\mathrm{b}}$ taken as $\left(  x_{0}-x_{\mathrm{b}}\right)
^{\mathrm{T}}\boldsymbol{B}^{-1}\left(  x_{0}-x_{\mathrm{b}}\right)  \approx
0$, etc., we get the nonlinear least squares problem%
\begin{align}
&  \left\Vert x_{0}-x_{\mathrm{b}}\right\Vert _{\boldsymbol{B}^{-1}}^{2}+%
{\displaystyle\sum\limits_{i=1}^{k}}
\left\Vert x_{i}-\mathcal{M}_{i}\left(  x_{i-1}\right)  \right\Vert
_{\boldsymbol{Q}_{i}^{-1}}^{2}\nonumber\\
&  +%
{\displaystyle\sum\limits_{i=1}^{k}}
\left\Vert y_{i}-\mathcal{H}_{i}\left(  x_{i}\right)  \right\Vert
_{\boldsymbol{R}_{i}^{-1}}^{2}\rightarrow\min_{x_{0:k}}, \label{eq:4dvar}%
\end{align}
called weak-constraint 4DVAR \citep{Tremolet-2007-ME4}. Originally in 4DVAR,
$x_{i}=$ $\mathcal{M}_{i}\left(  x_{i-1}\right)  $; the weak constraint
$x_{i}\approx\mathcal{M}_{i}\left(  x_{i-1}\right)  $ accounts for model error.

The least squares problem (\ref{eq:4dvar}) is solved iteratively by
linearization,%
\begin{align*}
\mathcal{M}_{i}\left(  x_{i-1}+\delta x_{i-1}\right)   &  \approx
\mathcal{M}_{i}\left(  x_{i-1}\right)  +\mathcal{M}_{i}^{\prime}\left(
x_{i-1}\right)  \delta x_{i-1},\\
\mathcal{H}_{i}\left(  x_{i}+\delta x_{i}\right)   &  \approx\mathcal{H}%
_{i}\left(  x_{i}\right)  +\mathcal{H}_{i}^{\prime}\left(  x_{i}\right)
\delta x_{i} .
\end{align*}
For $k$ vectors $u_{i}$, $i=1\dots k$, denote the composite vector
\[
u_{0:k}=\left[
\begin{array}
[c]{c}%
u_{0}\\
\vdots\\
u_{k}%
\end{array}
\right]  .
\]

In each iteration $x_{0:k}\leftarrow x_{0:k}+\delta x_{0:k}$, one solves the
auxiliary linear least squares problem for the increments $\delta x_{0:k},$%
\begin{align}
&  \left\Vert x_{0}+\delta x_{0}-x_{\mathrm{b}}\right\Vert _{\boldsymbol{B}%
^{-1}}^{2}\nonumber\\
&  +%
{\displaystyle\sum\limits_{i=1}^{k}}
\left\Vert x_{i}+\delta x_{i}-\mathcal{M}_{i}\left(  x_{i-1}\right)
-\mathcal{M}_{i}^{\prime}\left(  x_{i-1}\right)  \delta x_{i-1}\right\Vert
_{\boldsymbol{Q}_{i}^{-1}}^{2}\nonumber\\
&  +%
{\displaystyle\sum\limits_{i=1}^{k}}
\left\Vert y_{i}-\mathcal{H}_{i}\left(  x_{i}\right)  -\mathcal{H}_{i}%
^{\prime}\left(  x_{i}\right)  \delta x_{i}\right\Vert _{\boldsymbol{R}%
_{i}^{-1}}^{2}\rightarrow\min_{\delta x_{0:k}}. \label{eq:4dvar-incr}%
\end{align}
This is the Gauss-Newton method \citep{Bell-1994-IKS,Tshimanga-2008-LPA} for
nonlinear squares, known in 4DVAR as the incremental approach \citep{Courtier-1994-SOI}.

Denote
\begin{align}
z_{0:k}  &  =\delta x_{0:k},\quad z_{\mathrm{b}}=x_{\mathrm{b}}-x_{0}%
,\nonumber\\
m_{i}  &  =\mathcal{M}_{i}\left(  x_{i-1}\right)  -x_{i},\quad d_{i}%
=y_{i}-\mathcal{H}_{i}\left(  x_{i}\right)  ,\label{eq:def-di}\\
\boldsymbol{M}_{i}  &  =\mathcal{M}_{i}^{\prime}\left(  x_{i-1}\right)
,\quad\boldsymbol{H}_{i}=\mathcal{H}_{i}^{\prime}\left(  x_{i}\right)
\nonumber
\end{align}
and write the auxiliary linear least squares problem (\ref{eq:4dvar-incr}) as%
\begin{align}
&  \left\Vert z_{0}-z_{\mathrm{b}}\right\Vert _{\boldsymbol{B}^{-1}}^{2}+%
{\displaystyle\sum\limits_{i=1}^{k}}
\left\Vert z_{i}-\boldsymbol{M}_{i}z_{i-1}-m_{i}\right\Vert _{\boldsymbol{Q}%
_{i}^{-1}}^{2}\nonumber\\
&  +%
{\displaystyle\sum\limits_{i=1}^{k}}
\left\Vert d_{i}-\boldsymbol{H}_{i}z_{i}\right\Vert _{\boldsymbol{R}_{i}^{-1}%
}^{2}\rightarrow\min_{z_{0:k}} \label{eq:4dvar-lin}%
\end{align}
\qquad

The function minimized in (\ref{eq:4dvar-lin}) is exactly the same as the one
minimized in the Kalman smoother. The Gauss-Newton method with the Kalman
smoother as the linear least squares solver is known as the iterated Kalman
smoother, and considerable improvements can be obtained against running the
Kalman smoother, applied to the linearized problem, only once
\citep{Bell-1994-IKS,Fisher-2005-EKS}%
.

\section{Ensemble Kalman Filter and Smoother}

\label{sec:EnKS}

We present the EnKF and EnKS algorithms, essentially following
\cite{Evensen-2009-DAE}, in a form needed to state our theorems. We start with
a formulation of the EnKF in a notation suitable for extension to EnKS. The
notation $v^{\ell}\sim N\left(  m,\boldsymbol{A}\right)  $ means that
$v^{\ell}$ is sampled from $N\left(  m,\boldsymbol{A}\right)  $ independently
of anything else. The ensemble of states of the linearized model at time $i$,
conditioned on data up to time $j$ (that is, with the data up to time $j$
already ingested), is denoted by
\[
Z_{i|j}^{N}=\left[  z_{i|j}^{1},\ldots,z_{i|j}^{N}\right]  =\left[
z_{i|j}^{\ell}\right]  ,
\]
where the ensemble member index $\ell$ always runs over $\ell=1,\ldots,N$, and
similarly for other ensembles.

\begin{algorithm}
[EnKF]1. Initialize\textbf{ }%
\begin{equation}
z_{0|0}^{\ell}\sim N\left(  z_{\mathrm{b}},\boldsymbol{B}\right)  ,\quad
\ell=1,\ldots,N. \label{eq:enks-init}%
\end{equation}

2. For $i=1,\ldots,k$, advance in time%
\begin{equation}
z_{i|i-1}^{\ell}=\boldsymbol{M}_{i}z_{i-1|i-1}^{\ell}+m_{i}+v_{i}^{\ell},\quad
v_{i}^{\ell}\sim N\left(  0,\boldsymbol{Q}_{i}\right)  , \label{eq:advance}%
\end{equation}
followed by the analysis step%
\begin{align}
z_{i|i}^{\ell}=  &  z_{i|i-1}^{\ell}-\boldsymbol{P}_{i}^{N}\boldsymbol{H}%
_{i}^{\mathrm{T}}(\boldsymbol{H}_{i}\boldsymbol{P}_{i}^{N}\boldsymbol{H}%
_{i}^{\mathrm{T}}+\boldsymbol{R}_{i})^{-1}\nonumber\\
&  \cdot(\boldsymbol{H}_{i}z_{i|i-1}^{\ell}-d_{i}-w_{i}^{\ell}),\quad
w_{i}^{\ell}\sim N\left(  0,\boldsymbol{R}_{i}\right)  , \label{eq:analysis}%
\end{align}
where
\begin{equation}
\boldsymbol{P}_{i}^{N}=\frac{1}{N-1}(Z_{i|i-1}^{N}-\overline{z}_{i|i-1}%
^{N}\boldsymbol{1}^{\mathrm{T}})(Z_{i|i-1}^{N}-\overline{z}_{i|i-1}%
^{N}\boldsymbol{1}^{\mathrm{T}})^{\mathrm{T}} \label{eq:cov}%
\end{equation}
is the sample covariance,
\[
\overline{z}_{i|i-1}^{N}=Z_{i|i-1}^{N}\frac{\boldsymbol{1}}{N}%
\]
is the sample mean, and $\boldsymbol{1}$ is the vector of all ones size
$N\times1$.
\end{algorithm}

\begin{remark}
From (\ref{eq:cov}),%
\begin{equation}
\boldsymbol{P}_{i}^{N}=Z_{i|i-1}^{N}\frac{1}{N-1}\left(  \boldsymbol{I}%
-\frac{\boldsymbol{11}^{\mathrm{T}}}{N}\right)  \left(  \boldsymbol{I}%
-\frac{\boldsymbol{11}^{\mathrm{T}}}{N}\right)  Z_{i|i-1}^{N\mathrm{T}},
\label{eq:cov-fact}%
\end{equation}
and
\begin{equation}
\boldsymbol{H}_{i}\boldsymbol{P}_{i}^{N}\boldsymbol{H}_{i}^{\mathrm{T}}%
=\frac{1}{N-1}A_{i}^{N}A_{i}^{N\mathrm{T}}. \label{eq:obs-cov-fact}%
\end{equation}

Hence, the analysis ensemble $Z_{i|i}^{N}$ consists of linear combination of
the forecast ensemble, which can be written as multiplying the ensemble by a
transformation matrix $\boldsymbol{T}_{i}^{N}$,%
\begin{equation}
Z_{i|i}^{N}=Z_{i|i-1}^{N}\boldsymbol{T}_{i}^{N},\quad\boldsymbol{T}_{i}^{N}%
\in\mathbb{R}^{N\times N}, \label{eq:enkf-transf}%
\end{equation}
where%
\begin{align}
\boldsymbol{T}_{i}^{N}  &  =I-\frac{1}{N-1}\left(  \boldsymbol{I}%
-\frac{\boldsymbol{11}^{\mathrm{T}}}{N}\right)  A_{i}^{N\mathrm{T}%
}\label{eq:EnKF-T}\\
&  \cdot\left(  \frac{1}{N-1}A_{i}^{N}A_{i}^{N\mathrm{T}}+\boldsymbol{R}%
_{i}\right)  ^{-1}\nonumber\\
&  \cdot\left[  \boldsymbol{H}_{i}z_{i|i-1}^{\ell}-d_{i}+w_{i}^{\ell}\right]
_{\ell=1,N}, \label{eq:incremental-Z-fd}%
\end{align}
with $w_{i}^{\ell}\sim N\left(  d_{i},\boldsymbol{R}_{i}\right)  $, and
\begin{align}
A_{i}^{N}  &  =\boldsymbol{H}_{i}Z_{i:i-1}^{N}\left(  I-\frac{\boldsymbol{11}%
^{\mathrm{T}}}{N}\right)  =\left[  a_{i}^{1},\ldots,a_{i}^{N}\right]
,\nonumber\\
a_{i}^{\ell}  &  =\boldsymbol{H}_{i}z_{i|i-1}^{\ell}-\frac{1}{N}%
{\displaystyle\sum\limits_{j=1}^{N}}
\boldsymbol{H}_{i}z_{i|i-1}^{j}. \label{eq:obs-ens-member}%
\end{align}

\end{remark}

\begin{remark}
The matrix formula in the analysis step (\ref{eq:analysis}) is not efficient
when the dimension of the data space is large. Using (\ref{eq:obs-cov-fact})
and the Sherman-Morrisson-Woodbury formula
\citep{Hager-1989-UIM}%
, we transform the inverse in (\ref{eq:analysis}) into
\begin{align}
&  \left(  \boldsymbol{H}_{i}\boldsymbol{P}_{i}^{N}\boldsymbol{H}%
_{i}^{\mathrm{T}}+\boldsymbol{R}_{i}\right)  ^{-1}=\boldsymbol{R}_{i}%
^{-1}\nonumber\\
&  \quad\cdot\left[  \boldsymbol{I}-\frac{1}{N-1}A_{i}^{N}\left(
\boldsymbol{I}+\frac{A_{i}^{N\mathrm{T}}\boldsymbol{R}_{i}^{-1}A_{i}^{N}}%
{N-1}\right)  ^{-1}A_{i}^{N\mathrm{T}}\boldsymbol{R}_{i}^{-1}\right]  .
\label{eq:inverse-eff}%
\end{align}
Using (\ref{eq:inverse-eff}) requires only the solution of systems with the
data error covariance matrix $\boldsymbol{R}_{i}$ (which is typically easy,
and often $\boldsymbol{R}_{i}$ is diagonal) and of a system of the size $N$,
the number of ensemble members. See \cite{Mandel-2009-DAW} for details and
operation counts.\textbf{ }
\end{remark}

The EnKS is obtained by applying the same analysis step (\ref{eq:analysis}) as
in the EnKF to the composite state $Z_{0:i|i-1}$from time $0$ to $i$,
conditioned on data up to time $i-1,$%
\[
Z_{0:i|i-1}^{N}=\left[
\begin{array}
[c]{c}%
Z_{0|i-1}^{N}\\
\vdots\\
Z_{i|i-1}^{N}%
\end{array}
\right]  .
\]
in the place of $Z_{i|i-1}$. The observation term $\boldsymbol{H}_{i}%
Z_{i|i-1}^{N}-D_{i}$ becomes
\begin{equation}
\left[  0,\ldots,\boldsymbol{H}_{i}\right]  Z_{0:i|i-1}^{N}-D_{i}%
=\boldsymbol{H}_{i}Z_{i|i-1}^{N}-D_{i}. \label{eq:obs-comp}%
\end{equation}

\begin{algorithm}
[EnKS]\label{alg:EnKS} Given $z_{\mathrm{b}}$,

1. Initialize\textbf{ }%
\begin{equation}
z_{0|0}^{\ell}\sim N\left(  z_{\mathrm{b}},\boldsymbol{B}\right)  ,\quad
\ell=1,\ldots,N. \label{eq:EnKS-init}%
\end{equation}

2. For $i=1,\ldots,k$, advance in time,%
\begin{equation}
z_{i|i-1}^{\ell}=\boldsymbol{M}_{i}z_{i-1|i-1}^{\ell}+m_{i}+v_{i}^{\ell},\quad
v_{i}^{\ell}\sim N\left(  0,\boldsymbol{Q}_{i}\right)  ,
\label{eq:EnKS-advance}%
\end{equation}
followed by the analysis step%
\begin{align}
Z_{0:i|i}^{N}=  &  Z_{0:i|i-1}^{N}-\boldsymbol{P}_{0:i,0:i}^{N}%
\widetilde{\boldsymbol{H}}_{0:i}^{\mathrm{T}}(\widetilde{\boldsymbol{H}}%
_{0:i}\boldsymbol{P}_{0:i,0:i}\widetilde{\boldsymbol{H}}_{0:i}^{\mathrm{T}%
}+\boldsymbol{R}_{i})^{-1}\label{eq:smoother-analysis}\\
&  \cdot(\widetilde{\boldsymbol{H}}_{0:i}Z_{i|i-1}^{N}-D_{i}),\quad D_{i}\sim
N\left(  d_{i},\boldsymbol{R}_{i}\right)  ,\nonumber
\end{align}
where $\widetilde{\boldsymbol{H}}_{0:i}=\left[  0,\ldots,\boldsymbol{H}%
_{i}\right]  $, and $\boldsymbol{P}_{0:i,0:i}^{N}$ is the sample covariance
matrix of $Z_{0:i|i-1}^{N}$.
\end{algorithm}

The following theorem allows a straightforward implementation of the EnKS from
the EnKF -- the same transformation matrix is applied to the composite state
from times $0$ to $i$, not just the last time $i$. Also, one can use a
transformation matrix from another version of EnKF, such as the square root
filter, e.g., LETKF
\citep{Hunt-2007-EDA}%
; \cite{Fertig-2007-CS4} assume such relation a-priori for a related method
based on LETKF.

\begin{theorem}
\label{thm:T}The EnKS satisfies%
\begin{equation}
Z_{0:i|i}^{N}=Z_{0:i|i-1}^{N}\boldsymbol{T}_{i}^{N}. \label{eq:enks-trans}%
\end{equation}
where $\boldsymbol{T}_{i}^{N}$ is the transformation matrix
(\ref{eq:enkf-transf}) from the EnKF.
\end{theorem}

\begin{proof}
We have
\[
\boldsymbol{P}_{0:i,0:i}^{N}=\left[
\begin{array}
[c]{lll}%
\boldsymbol{P}_{0,0}^{N} & \cdots & \boldsymbol{P}_{0,i}^{N}\\
\vdots & \ddots & \vdots\\
\boldsymbol{P}_{i,0}^{N} & \cdots & \boldsymbol{P}_{i,i}^{N}%
\end{array}
\right]  ,
\]
with the blocks
\[
\boldsymbol{P}_{j\ell}^{N}=\frac{1}{N-1}(Z_{j|i-1}^{N}-\overline{z}%
_{j|i-1}\boldsymbol{1}^{\mathrm{T}})(Z_{\ell|i-1}^{N}-\overline{z}_{\ell
|i-1}\boldsymbol{1}^{\mathrm{T}})^{\mathrm{T}}.
\]
The terms in (\ref{eq:analysis}) become in (\ref{eq:smoother-analysis})%
\[
\boldsymbol{P}_{0:i,0:i}^{N}\widetilde{\boldsymbol{H}}^{\mathrm{T}%
}=\boldsymbol{P}_{0:i,0:i}^{N}\left[
\begin{array}
[c]{c}%
0\\
\vdots\\
\boldsymbol{H}_{i}^{\mathrm{T}}%
\end{array}
\right]  =\boldsymbol{P}_{0:i,i}^{N}\boldsymbol{H}_{i}^{\mathrm{T}},
\]%
\[
\widetilde{\boldsymbol{H}}^{\mathrm{T}}\boldsymbol{P}_{0:i,0:i}^{N}%
\widetilde{\boldsymbol{H}}=\boldsymbol{H}_{i}\boldsymbol{P}_{i,i}%
^{N}\boldsymbol{H}_{i}^{\mathrm{T}},
\]
and, similarly as in (\ref{eq:cov-fact}),%
\begin{equation}
\boldsymbol{P}_{0:i,i}^{N}=Z_{0:i|i-1}^{N}\frac{1}{N-1}\left(  \boldsymbol{I}%
-\frac{\boldsymbol{11}^{\mathrm{T}}}{N}\right)  \left(  \boldsymbol{I}%
-\frac{\boldsymbol{11}^{\mathrm{T}}}{N}\right)  Z_{i|i-1}^{N\mathrm{T}}.
\label{eq:cov-comp}%
\end{equation}

The result now follows by the comparison of (\ref{eq:obs-comp}%
)--(\ref{eq:cov-comp}) with (\ref{eq:analysis})--(\ref{eq:enkf-transf}).
\end{proof}

When the original, nonlinear operators instead of the linearizations are used,
we get the nonlinear EnKS method, which is common and useful in practice, even
if it may not be justified theoretically. This method is obtained from the
linear EnKS by replacing (\ref{eq:EnKS-advance}) and (\ref{eq:obs-ens-member})
by their original, nonlinear versions. It operates on the original ensemble of
the states $X^{N}=\left[  x^{\ell}\right]  _{\ell=1}^{N}$ rather than on the
increments $z^{\ell}=\delta x^{\ell}$.

\begin{algorithm}
[Nonlinear EnKS]\label{alg:nonlinear-EnKS}

1. Initialize\textbf{ }%
\[
x_{0|0}^{k}\sim N\left(  x_{\mathrm{b}},\boldsymbol{B}\right)  .
\]

2. For $i=1,\ldots,k$, advance in time%
\begin{equation}
x_{i|i-1}^{\ell}=\mathcal{M}_{i}\left(  x_{i-1|i-1}^{\ell}\right)
+v_{i},\quad v_{i}\sim N\left(  0,\boldsymbol{Q}_{i}\right)
\label{eq:nonlinear-EnKS-advance}%
\end{equation}
followed by the analysis step%
\[
X_{0:i|i}^{N}=X_{0:i|i-1}^{N}\boldsymbol{T}_{i}^{N},\quad\boldsymbol{T}%
_{i}^{N}\in\mathbb{R}^{N\times N},
\]
where%
\begin{align}
\boldsymbol{T}_{i}^{N}  &  =I-\frac{1}{N-1}\left(  \boldsymbol{I}%
-\frac{\boldsymbol{11}^{\mathrm{T}}}{N}\right)  A_{i}^{N\mathrm{T}}\nonumber\\
&  \cdot\left(  \frac{1}{N-1}A_{i}^{N}A_{i}^{N\mathrm{T}}+\boldsymbol{R}%
_{i}\right)  ^{-1}\nonumber\\
&  \cdot\left[  \mathcal{H}_{i}\left(  x_{i|i-1}^{\ell}\right)  -y_{i}%
-w_{i}^{\ell}\right]  _{\ell=1,N},\quad w_{i}^{\ell}\sim N\left(
0,\boldsymbol{R}_{i}\right)  ,\label{eq:eval-innovation}\\
A_{i}^{N}  &  =\left[  a_{i}^{1},\ldots,a_{i}^{N}\right]  ,\nonumber\\
a_{i}^{\ell}  &  =\mathcal{H}_{i}\left(  x_{i|i-1}^{\ell}\right)  -\frac{1}{N}%
{\displaystyle\sum\limits_{j=1}^{N}}
\mathcal{H}_{i}\left(  x_{i|i-1}^{j}\right)  , \label{eq:eval-Z-from-h}%
\end{align}
and $\overline{x}_{i|i-1}=X_{i|i-1}^{N}\boldsymbol{1}/N.$
\end{algorithm}

Using Theorem \ref{thm:T}, it is easy to see that Algorithm
\ref{alg:nonlinear-EnKS} coincides with Algorithm \ref{alg:EnKS} in the linear
case, i.e., when $\mathcal{M}_{i}$ and $\mathcal{H}_{i}$ are affine operators.

\section{Nonlinear EnKS-4DVAR method}

\label{sec:FD}

So far, the algorithm was relying on the linearized (i.e., tangent) model
operators $\boldsymbol{M}_{i}$ and $\boldsymbol{H}_{i}$ and their adjoints.

The linearized model $\boldsymbol{M}_{i}=\mathcal{M}_{i}^{\prime}\left(
x_{i-1}\right)  $ occurs only in advancing the time as action on the ensemble
members $\delta x^{\ell}=z^{\ell}$,%
\[
\boldsymbol{M}_{i}z_{i-1}^{\ell}+m_{i}=\mathcal{M}_{i}^{\prime}\left(
x_{i-1}\right)  z_{i-1}^{\ell}+\mathcal{M}_{i}\left(  x_{i-1}\right)  -x_{i}%
\]
Approximating by finite differences based at $x_{i-1}$ with step $\tau>0$, we
get%
\begin{align}
\boldsymbol{M}_{i}z_{i-1}^{\ell}+m_{i}  &  \approx\frac{\mathcal{M}_{i}\left(
x_{i-1}+\tau z_{i-1}^{\ell}\right)  -\mathcal{M}_{i}\left(  x_{i-1}\right)
}{\tau}\label{eq:tangent-M}\\
&  +\mathcal{M}_{i}\left(  x_{i-1}\right)  -x_{i}.\nonumber
\end{align}
Thus, advancing the linarized model in time requires $N+1$ evaluations of
$\mathcal{M}_{i}$, at $x_{i-1}$ and $x_{i-1}+\tau\delta x_{i-1}^{n}$.

The observation matrix $\boldsymbol{H}_{i}$ occurs only in the action on the
ensemble,%
\[
\boldsymbol{H}_{i}Z^{N}=\left[  \boldsymbol{H}_{i}z_{i}^{1},\ldots
,\boldsymbol{H}_{i}z_{i}^{N}\right]  .
\]
Approximating by finite differences based at $x_{i}$, with step $\tau>0$, we
have%
\begin{equation}
\boldsymbol{H}z_{i}^{\ell}\approx\frac{\mathcal{H}_{i}\left(  x_{i-1}+\tau
z_{i}^{\ell}\right)  -\mathcal{H}_{i}\left(  x_{i-1}\right)  }{\tau}.
\label{eq:tangent-H}%
\end{equation}
Thus, evaluating the action of the linarized observation requires $N+1$
evaluations of $\mathcal{H}_{i}$, at $x_{i-1}$ and $x_{i-1}+\tau z_{i-1}%
^{\ell}$.

Here is the overall method. First, initialize%
\[
x_{0}=x_{\mathrm{b}},\quad x_{i}=M_{i}\left(  x_{i-1}\right)  ,i=1,\ldots,k,
\]
if not given already. One iteration (\ref{eq:4dvar-incr}) of the incremental
4DVAR is then implemented as follows.

\begin{algorithm}
[EnKS-4DVAR]\label{alg:EnKS-4DVAR}Given $x_{0},\ldots,x_{k}$:

1. Initialize $z_{0|0}^{\ell}\sim N\left(  z_{\mathrm{b}},\boldsymbol{B}%
\right)  $ following (\ref{eq:enks-init}), with $z_{\mathrm{b}}=0$.

2. For $i=1,\ldots,k$, advance $z^{\ell}$in time following
(\ref{eq:EnKS-advance}) with the linearized operator approximated from
(\ref{eq:tangent-M}),%
\begin{align}
z_{i|i-1}^{\ell}=  &  \frac{\mathcal{M}_{i}\left(  x_{i-1}+\tau z_{i-1|i-1}%
^{\ell}\right)  -\mathcal{M}_{i}\left(  x_{i-1}\right)  }{\tau}%
\label{eq:incremental-smoothing-adv-fd}\\
&  +\mathcal{M}_{i}\left(  x_{i-1}\right)  -x_{i}+v_{i}^{\ell},\quad
v_{i}^{\ell}\sim N\left(  0,\boldsymbol{Q}_{i}\right)  ,\nonumber
\end{align}
followed by the analysis step (\ref{eq:enks-trans}) with the transformation
matrix $\boldsymbol{T}_{i}^{N}$ computed from (\ref{eq:EnKF-T}) and with the
matrix-vector products $\boldsymbol{H}_{i}z_{i}$ approximated from
(\ref{eq:tangent-H}).

3. Update
\[
x_{i}\longleftarrow x_{i}+\frac{1}{N}%
{\displaystyle\sum\limits_{\ell=1}^{N}}
z_{i|k}^{\ell},\quad i=0,\ldots,k.
\]

\end{algorithm}

Note that for small $\tau$, the resulting method is asymptotically equivalent
to the method with the derivatives. Amazingly, it turns out that in the case
when $\tau=1$, we recover the standard EnKS applied directly to the nonlinear
problems, that is, with the linearized advance in time (\ref{eq:advance})
replaced by application of the original, nonlinear operator $\mathcal{M}_{i}$.
In particular, the incremental 4DVAR does not converge unless it is already at
a stationary point, because each iteration delivers the same result, up to the
randomness of the EnKS.

\begin{theorem}
\label{thm:tau1}If $\tau=1$, then one step of EnKS-4DVAR (Algorithm
\ref{alg:EnKS-4DVAR}) is exactly the nonlinear EnKS (Algorithm
\ref{alg:nonlinear-EnKS}). In particular, the result of the step does not
depend on the previous iterate.
\end{theorem}

\begin{proof}
Indeed, (\ref{eq:incremental-smoothing-adv-fd}) becomes%
\begin{align*}
z_{i|i-1}^{\ell}  &  =\frac{\mathcal{M}_{i}\left(  x_{i-1}^{{}}+z_{i-1|i-1}%
^{\ell}\right)  -\mathcal{M}_{i}\left(  x_{i-1}^{{}}\right)  }{1}\\
&  +\mathcal{M}_{i}\left(  x_{i-1}^{{}}\right)  -x_{i}^{{}}+v_{i}^{\ell}\\
&  =\mathcal{M}_{i}\left(  x_{i-1}+z_{i-1|i-1}^{\ell}\right)  -x_{i}^{{}%
}+v_{i}^{\ell},
\end{align*}
hence%
\begin{equation}
x_{i}+z_{i|i-1}^{\ell}=\mathcal{M}_{i}\left(  x_{i-1}+z_{i-1|i-1}^{\ell
}\right)  +v_{i}^{\ell} \label{eq:nonlinear-adv}%
\end{equation}
which is exactly the same as advancing the ensemble member $\ell$ following
(\ref{eq:nonlinear-EnKS-advance}) with $x_{i-1|i-1}^{\ell}=x_{i-1}%
+z_{i-1|i-1}^{\ell}$. Similarly, (\ref{eq:obs-ens-member}) becomes with
$\tau=1$
\begin{align}
a_{i}^{\ell}=  &  \frac{\mathcal{H}_{i}\left(  x_{i}+z_{i|i-1}^{\ell}\right)
-\mathcal{H}_{i}\left(  x_{i}\right)  }{1}\nonumber\\
&  -\frac{1}{N}\sum_{j=1}^{N}\frac{\mathcal{H}_{i}\left(  x_{i}+z_{i|i-1}%
^{j}\right)  -\mathcal{H}_{i}\left(  x_{i}\right)  }{1}\\
=  &  \mathcal{H}_{i}\left(  x_{i}+z_{i|i-1}^{\ell}\right)  -\frac{1}{N}%
\sum_{j=1}^{N}\mathcal{H}_{i}\left(  x_{i}+z_{i|i-1}^{j}\right)  ,
\label{eq:nonlinear-Z}%
\end{align}
which is exactly the same as (\ref{eq:eval-Z-from-h}) with $x_{i|i-1}^{\ell
}=x_{i}+z_{i|i-1}^{\ell}$. Finally, (\ref{eq:incremental-Z-fd}) becomes using
(\ref{eq:def-di}),%
\begin{align}
&  \boldsymbol{H}_{i}z_{i|i-1}^{\ell}-d_{i}\nonumber\\
&  =\frac{\mathcal{H}_{i}\left(  x_{i}+z_{i|i-1}^{\ell}\right)  -\mathcal{H}%
_{i}\left(  x_{i}\right)  }{1}-\left[  y_{i}-\mathcal{H}_{i}\left(
x_{i}\right)  \right] \\
&  =\mathcal{H}_{i}\left(  x_{i|i-1}^{\ell}\right)  -y_{i}
\label{eq:nonlinear-D}%
\end{align}
which is exactly the same as in (\ref{eq:eval-innovation}).
\end{proof}

\section{Tikhonov regularization and the Levenberg-Marquardt method}

\label{sec:reg}

The Gauss-Newton method may diverge, but convergence to a stationary point of
(\ref{eq:4dvar}) can be recovered by a control of the step $\delta x$. Adding
a constraint of the form $\left\Vert \delta x_{i}\right\Vert \leq\varepsilon$
leads to globally convergent trust region methods
\citep{Gratton-2013-PGC}%
. Here, we add $\delta x_{i}$ in a Tikhonov regularization term of the form
$\gamma\left\Vert \delta x_{i}\right\Vert _{\boldsymbol{S}_{i}^{-1}}^{2}$,
which controls the step size as well as rotates the step direction towards the
steepest descent, and obtain the Levenberg-Marquardt method $x_{0:k}\leftarrow
x_{0:k}+\delta x_{0:k}$, where%
\begin{align}
&  \left\Vert \delta x_{0}-z_{\mathrm{b}}\right\Vert _{\boldsymbol{B}^{-1}%
}^{2}+%
{\displaystyle\sum\limits_{i=1}^{k}}
\left\Vert \delta x_{i}-\boldsymbol{M}_{i}\delta x_{i-1}-m_{i}\right\Vert
_{\boldsymbol{Q}_{i}^{-1}}^{2}\nonumber\\
&  \quad+%
{\displaystyle\sum\limits_{i=1}^{k}}
\left\Vert d_{i}-\boldsymbol{H}_{i}\delta x_{i}\right\Vert _{\boldsymbol{R}%
_{i}^{-1}}^{2}\nonumber\\
&  \quad+\gamma%
{\displaystyle\sum\limits_{i=0}^{k}}
\left\Vert \delta x_{i}\right\Vert _{\boldsymbol{S}_{i}^{-1}}^{2}%
\rightarrow\min_{\delta x_{0:k}} \label{eq:LM1}%
\end{align}

Under suitable technical assumptions, the Levenberg-Marquardt method is
guaranteed to converge globally if the regularization parameter $\gamma\geq0$
is large enough
\citep{Gill-1978-ASN,Osborne-1976-NLL}%
. Estimates for the convergence of the Levenberg-Marquardt method in the case
when the linear system is solved only approximately exist
\citep{Wright-1985-ILM}%
.

Similarly as in \cite{Johns-2008-TEK}, we interpret the regularization terms
$\gamma\left\Vert \delta x_{i}\right\Vert _{\boldsymbol{S}_{i}^{-1}}^{2}$ in
(\ref{eq:LM1}) as arising from additional independent observations $\delta
x_{i}\sim N\left(  0,\gamma^{-1}\boldsymbol{S}_{i}\right)  $ Because the
additional regularization observations $\delta x_{i}\approx0$ are independent
of the other observations and the state, separately, resulting in the
mathematically equivalent but often more efficient two-stage method - simply
run the EnKF analysis (\ref{eq:analysis})\ twice, and apply both
transformation matrices in turn following (\ref{eq:enks-trans}). With the
choice of $\boldsymbol{S}_{i}$ as identity or, more generally a diagonal
matrix, the implementation following (\ref{eq:inverse-eff}) is efficient; see
\cite{Mandel-2009-DAW} for operation counts.

Note that unlike in \cite{Johns-2008-TEK}, where the regularization was
applied to a nonlinear problem and thus the sequential data assimilation was
only approximate, here the EnKS is run on the auxiliar linearized problem
(\ref{eq:LM1}), so all distributions are gaussian and the equivalence of
solving (\ref{eq:LM1}) at once and assimilating the observations sequentially
is statistically exact.

\section{Computational results}

\label{sec:comp}

\subsection{Lorenz 63 model}

We first show an example without model error, where convergence is achieved
with $\gamma=0$(algo \ref{alg:EnKS-4DVAR}).

We consider the Lorenz 64 equations
\citep{Lorenz-1963-DNF-x}%
, a simple dynamical model with chaotic behaviour. The Lorenz equations are
given by the nonlinear system
\begin{align*}
\frac{dx}{dt}  &  =-\sigma(x-y)\\
\frac{dy}{dt}  &  =\rho x-y-xz\\
\frac{dz}{dt}  &  =xy-\beta z
\end{align*}
where $x=x(t)$, $y=y(t)$, $z=z(t)$ and $\sigma$, $\rho$, $\beta$ are
parameters, which in these experiments are chosen to have the values 10, 28
and 8/3 respectively. The system is discretized using the fourth-order
Runge-Kutta method. In (\ref{eq:4dvar}), we choose%
\[
\boldsymbol{B=}\sigma_{b}^{2}\operatorname{diag}\left(  1,\frac{1}{4},\frac
{1}{9}\right)  ,\quad R_{i}=\sigma_{r}^{2}\boldsymbol{I},
\]%
\[
\mathcal{H}_{i}\left(  x,y,z\right)  =\left(  x^{2},y^{2},z^{2}\right)  .
\]
In the experiments below, we assume perfect model, therefore $\boldsymbol{Q}%
_{i}=\varepsilon\boldsymbol{I}$, $\varepsilon\approx0$. We put $\sigma
_{b}=1,\sigma_{r}=1$, and $\epsilon=0.0001$.

\begin{figure}[ptb]
\begin{center}
\includegraphics[width=3.0in]{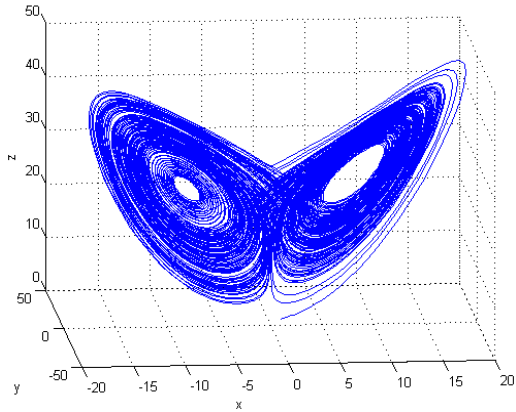}
\end{center}
\caption{The Lorenz attractor, initial values $x(0)=1$, $y(0)=1$, and
$z(0)=1$, time step $dt = 0.1$.}%
\label{fig:lorenz_63}%
\end{figure}

\begin{figure}[ptb]
\begin{center}
\hspace*{-0.4in} \includegraphics[width=3.5in]{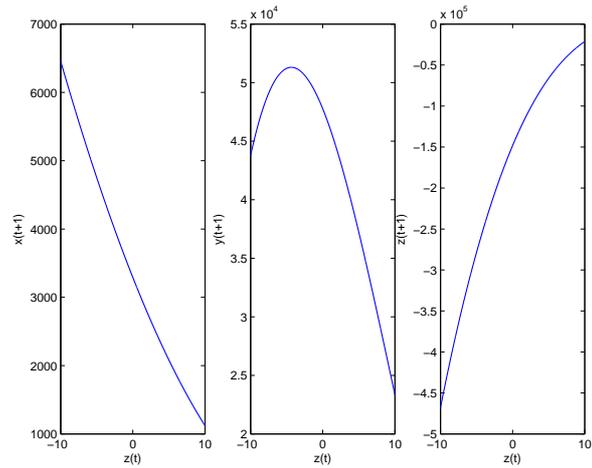}
\end{center}
\par
\vspace*{-0.2in}\caption{Nonlinearity of the Lorenz 63 model. The values of
$x(t+1), y(t+1)$ and $z(t+1)$ change quickly as a function of $x(t)=1$,
$y(t)=1$, and varying $z(t)$.}%
\label{fig:zperturbation}%
\end{figure}As can be seen in Figures \ref{fig:lorenz_63} and
\ref{fig:zperturbation}, the Lorenz attractor is fully nonlinear. It has two
lobes connected near the origin, and the trajectories of the system in this
saddle region are particularly sensitive to perturbations. Hence, slight
perturbations can alter the subsequent path from one lobe to the other. In
Figure \ref{fig:zperturbation}, we keep $x(t)$ and $y(t)$ constant and we vary
just $z(t)$, then we compute the state at time $t+1$, this figure shows the
non linear dependence between the different components of the state at time
$t+1$ and the third component of the state at time $t$,

\begin{table}[ptb]
\begin{center}
{\small
\begin{tabular}
[c]{|c|c|c|c|c|c|c|}\hline
Iteration & 1 & 2 & 3 & 4 & 5 & 6\\\hline
RMSE & 20.16 & 15.37 & 3.73 & 2.53 & 0.09 & 0.09\\\hline
\end{tabular}
}
\end{center}
\caption{Norm of the root mean square error in Gauss-Newton iterations with
EnKS as linear solver.}%
\label{tab:rmse_10_iterations}%
\end{table}

\begin{figure}[ptb]
\begin{center}
\includegraphics[width=2.5in]{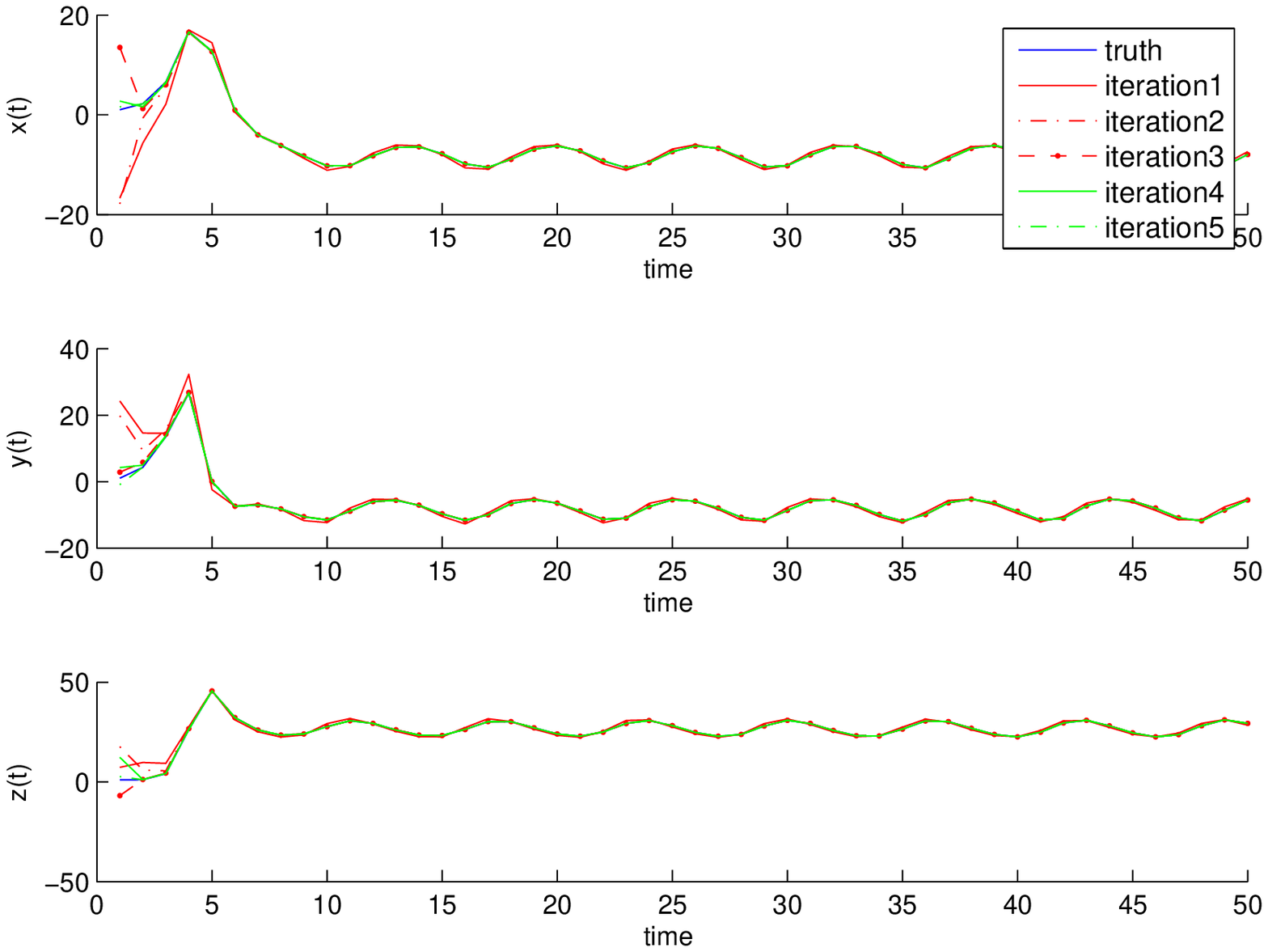}
\includegraphics[width=2.5in]{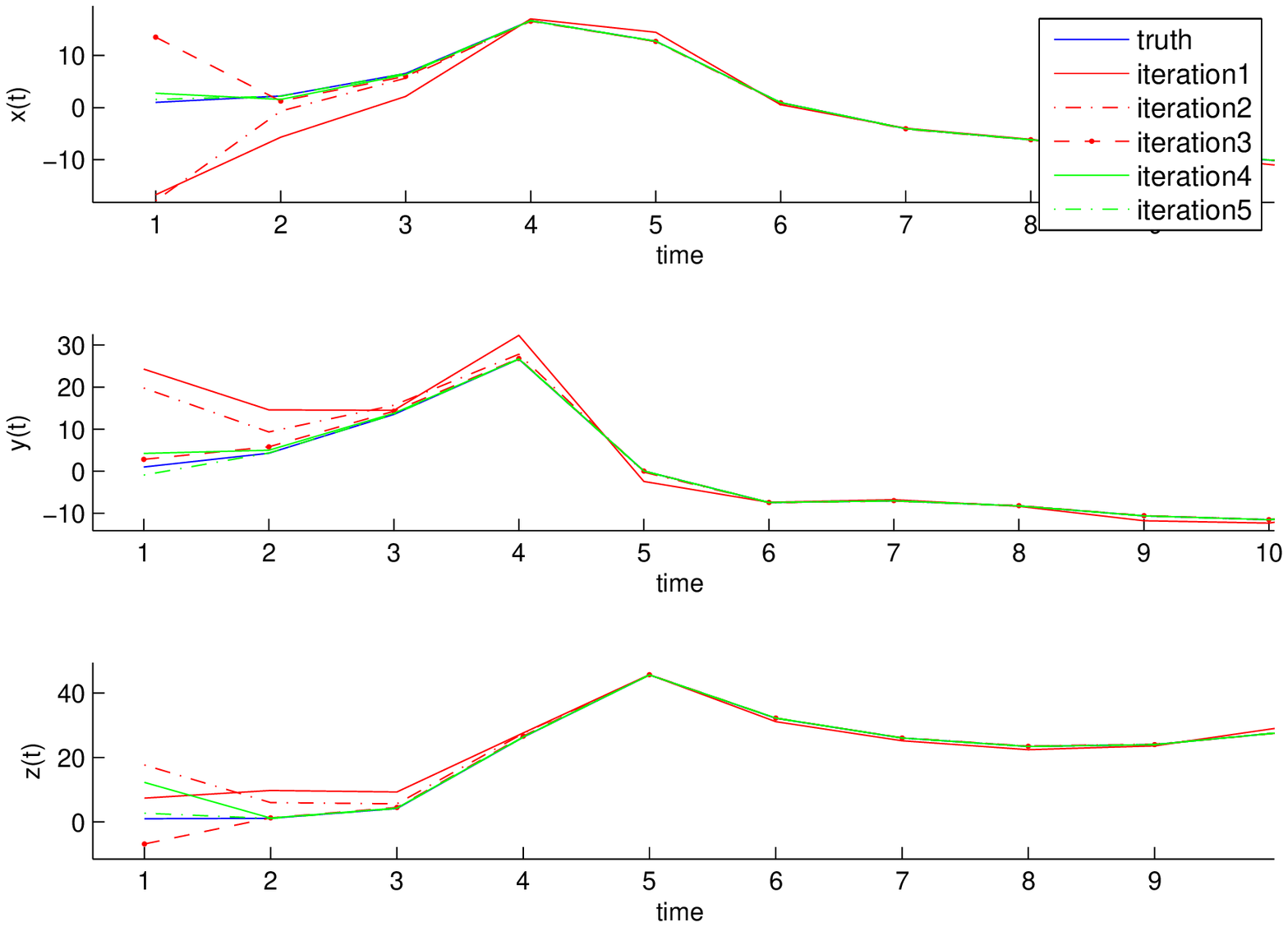}
\end{center}
\caption{The three components $x$, $y$, $z$ of the truth and five iterations
of EnKS-4DVAR. The initial conditions for the truth are $x(0)=1$, $y(0)=1$,
and $z(0)=1$, time step $dt = 0.1$, observations are the full state at each
time, ensemble size is $100$. }%
\label{fig:state_5_iterations}%
\end{figure}

\begin{figure}[ptb]
\begin{center}
\hspace*{0.4in}\includegraphics[width=2.2in]{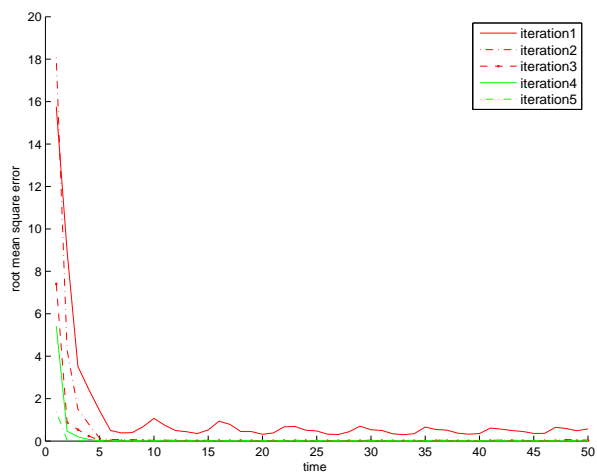}
\end{center}
\caption{Root mean square error of EnKS-4DVAR iterations. The problem setting
is the same as in Fig.~\ref{fig:state_5_iterations}.}%
\label{fig:rmse_5_iterations}%
\end{figure}

To evaluate the performance of the method, we use the twin experiment
technique. That is, an integration of the model is chosen as the true state.
We then obtain the data $y_{i}$ by applying the observation operator
$\mathcal{H}_{i}$ to the truth and then adding a gaussian perturbation
$N(0,\boldsymbol{R}_{i})$. Similarly, the background $x_{b}$ is sampled from
the gaussian distribution with the mean equal to the initial conditions and
the covariance $\boldsymbol{B}$. Then we try to recover the truth using the
observations $y_{i}$ and the background $x_{b}$.

Figure \ref{fig:state_5_iterations} reports simulation results for
assimilating observations over 50 assimilation cycles, using the Hybrid 4DVAR
and nonlinear EnKS method. Cycles are separated by a time interval of $dt =
0.1$. Figure \ref{fig:rmse_5_iterations} and shows the root mean square error
(RMSE) between the sample posterior mean and the true state of the system.

As can be seen from Table \ref{tab:rmse_10_iterations}, five iterations were
enough for the method to converge. Note that the error does not converge to
zero, because of the approximation and variability inherent in the ensemble approach.

\subsection{Lorenz 96 model}

We now consider the effect of the model error in the algorithm
\ref{alg:EnKS-4DVAR}.

\begin{figure}[ptb]
\begin{center}
\includegraphics[width=3.3in]{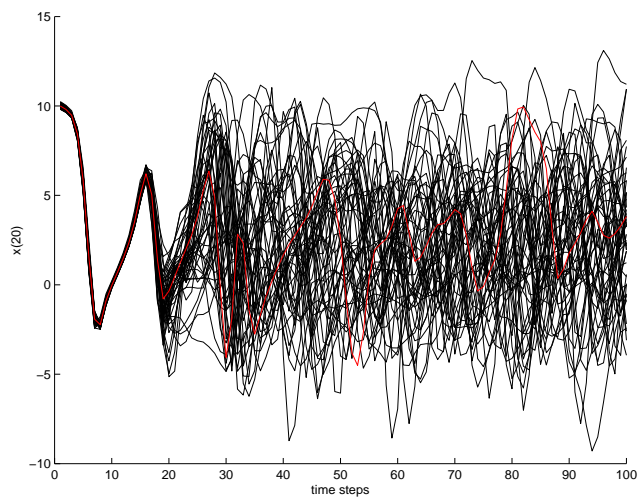}
\end{center}
\caption{Evolution of the 40-variable Lorenz-96 system at site 20. The red
line is unperturbed forecast. The black lines are an ensemble of 50 forecasts,
which start from slightly perturbed initial conditions.}%
\label{fig:lorenz96_perturbation}%
\end{figure}The Lorenz 96 model
\citep{Lorenz-2006-PPP}
is defined by the system of differential equations%
\[
\frac{dx_{j}}{dt}=\frac{1}{\kappa}(x_{j-1}(x_{j+1}-x_{j-2})-x_{j}+F),
\]
$j=1,\ldots,40$, with cyclic boundary conditions $x_{-1}=x_{39}$,
$x_{0}=x_{40}$, $x_{41}=x_{1}$. This model behaves chaotically in the case of
external forcing $F=8$. The first term of right-hand side simulates advection,
and this model can be regarded as the time evolution of a one-dimensional
quantity on a constant latitude circle, that is, the subscript corresponds to
longitude. The PDE is discretized using a fourth-order Runge-Kutta method. For
the results below $\kappa=1$, $F=8$ and $dt=0.01$. Figure
\ref{fig:lorenz96_perturbation} shows the chaotic dynamics of the Lorenz 96
system. For the tests we took the parameters $B=\sigma_{b}^{2}diag(1,...,\frac
{1}{i^{2}},...)$, $R_{i}=\sigma_{r}^{2}I$, $\mathcal{H}_{i}\left(
X_{i}\right)  =X_{t}+X_{t}^{2}$, $\mathcal{M}_{i}$ is the Lorenz 96 model,
$Q_{i}=\sigma_{q}C$, where $C_{i,j}=exp(-|i-j|\frac{dt}{L})$, $L=\frac{dt}{5}%
$, $\sigma_{b}=1,\sigma_{r}=1$, $\sigma_{q}=0.1$, and the ensemble size
$N=40$. We use again the twin experiments technique, the true state is equal
to an integration of the model plus a gaussian perturbation $N(0,Q_{i})$.
Figures \ref{fig:lorenz96_filter} and \ref{fig:rmse_lorenz96} illustrate
EnKS-4DVAR on this problem.

\begin{figure}[ptb]
\begin{center}
\includegraphics[width=3.3in]{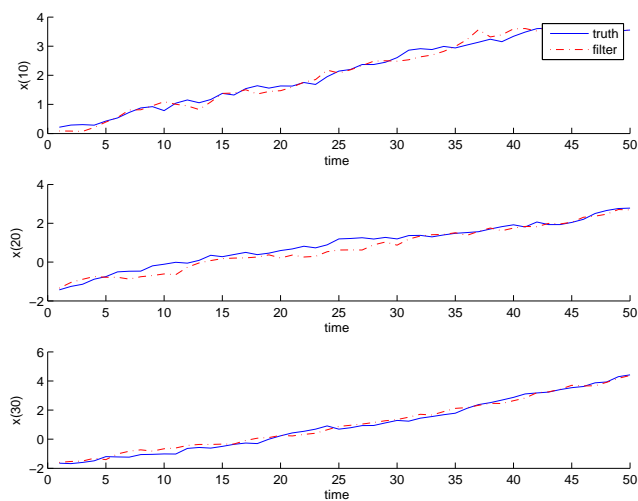}
\end{center}
\caption{EnKS-4DVAR for the Lorenz 96 system. The truth at $t=0$ is generated
randomly. The figure shows the truth and the filter by EnKS-4DVAR at
$x_{10}(t),$ $x_{20}(t),$ and $x_{30}(t)$. }%
\label{fig:lorenz96_filter}%
\end{figure}\begin{figure}[ptb]
\begin{center}
\includegraphics[width=3.3in]{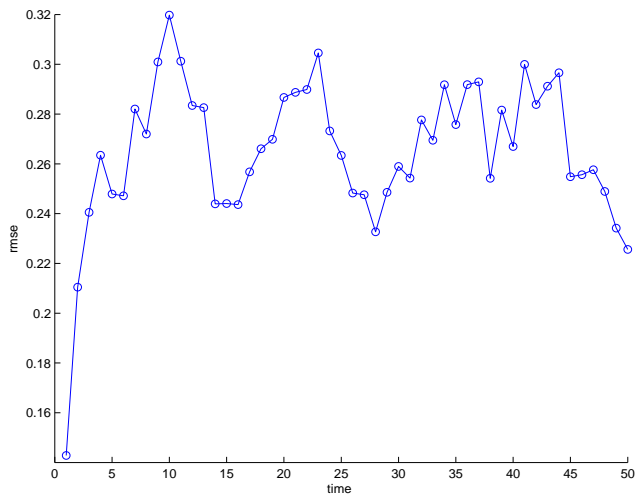}
\end{center}
\caption{Root mean square error between the truth and EnKS-4DVAR for the
Lorenz 96 system and the same setup as in Figure \ref{fig:lorenz96_filter}.}%
\label{fig:rmse_lorenz96}%
\end{figure}


\subsection{Example where Gauss-Newton does not converge}

We now show that the algorithm \ref{alg:EnKS-4DVAR} may not be convergent and
that Tikhonov regularization may be needed in some circumstances. The
following academic example illustrates this fact.

\begin{figure}[ptb]
\begin{center}
\includegraphics[width=3.35in]{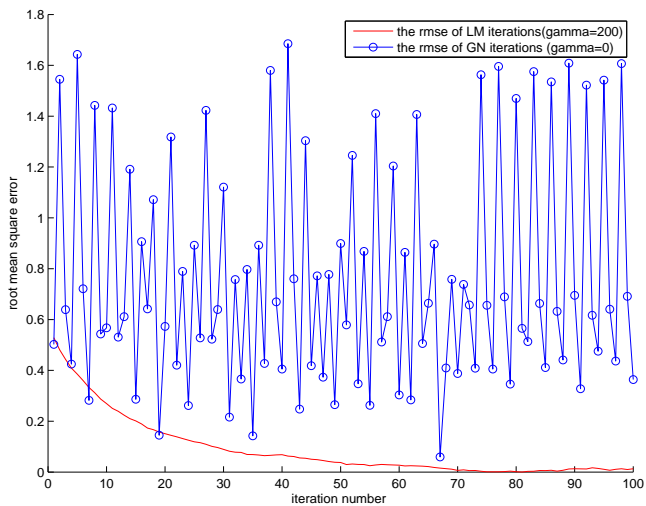}
\end{center}
\caption{The RMSE between Gauss-Newton iterations and local minimum
$(x_{0}^{\star},x_{1}^{\star})=(0.419,0.419)$($\gamma=0$, top), and the root
mean squared error between the Levenbert-Marquardt iteration and the local
minimum $(x_{0}^{\star},x_{1}^{\star})$($\gamma= 200$, bottom).}%
\label{fig:rmse_gn_LM}%
\end{figure}

The Gauss Newton method for nonlinear least squares is not globally
convergent, but convergence to a stationary point of any least square problem
can be recovered by using the Levenberg-Marquart control. Consider the
following example, which requires Levenberg-Marquart regularization to
converge. The objective function to minimize is
\begin{equation}
J(x_{0},x_{1})=(x_{0}-2)^{2}+(3+x_{1}^{3})^{2}+\frac{1}{q}(x_{0}-x_{1})^{2}
\label{eq:problem_rquiring_LM}%
\end{equation}
where $q$ is small, $q=0.000001$.

This problem could be seen as a 4DVAR problem where the state at time $0$ is
$x_{0}$, the background state is $x_{\mathrm{b}}=2$, the background covariance
$\boldsymbol{B}=\boldsymbol{I}$, there is one time step, the state at time
$t=1$ is $x_{1}$, the model $M_{1}=I$, and the model is perfect.
$Q_{1}=0.000001\approx0$, observation operator $\mathcal{H}_{1}(x)=-x^{3}$ and
observation error covariance is $\boldsymbol{R}_{1}=\boldsymbol{I}$.

Figure \ref{fig:rmse_gn_LM} shows the iterations of the EnKS-4DVAR method
applied to the problem \ref{eq:problem_rquiring_LM} seen as 4Dvar problem, in
two cases: when $\gamma=0$ ( Gauss Newton iteration) the method does not
converge, and for $\gamma=200$ ( Levenberg-Marquart iteration), the method
seems to converge to the local minimum $(x_{0}^{\star},x_{1}^{\star
})=(0.419,0.419)$.

\section{Conclusion}

\label{sec:conclusion}

The EnKS-4DVAR method was formulated and shown to be capable of handling
strongly nonlinear problems, and it converges to the nonlinear least squares
solution in a small number of iterations. Its performance on large and
realistic problems will be studied elsewhere.

\ack This research was supported by the Fondation STAE project ADTAO and the
National Science Foundation grants AGS-0835579 and DMS-1216481. A part of this
work was done when Jan Mandel was visiting INP-ENSEEIHT and CERFACS.

\bibliographystyle{wileyqj}
\bibliography{../../references/geo,../../references/other}

\end{document}